\documentclass[lettersize,journal]{IEEEtran}

\usepackage{authblk}
\usepackage[dvips]{graphicx}
\usepackage{epsfig}
\usepackage[cmex10]{amsmath}
\usepackage{amssymb}
\usepackage{amsthm}
\usepackage{amsfonts}
\usepackage{bm}
\usepackage{caption} 

\usepackage{cite}
\usepackage[top=0.7in,bottom=0.7in,left=0.625in,right=0.625in]{geometry}
\usepackage[svgnames]{xcolor}
\usepackage[tight,footnotesize]{subfigure}
\usepackage{algpseudocode}

\usepackage{algorithm}

\usepackage{tikz}
\usetikzlibrary{%
  arrows,%
  shapes.misc,% wg. rounded rectangle
  shapes.arrows,%
  chains,%
  matrix,%
  positioning,% wg. " of "
  scopes,%
  decorations.pathmorphing,% /pgf/decoration/random steps | erste Graphik
  shapes.geometric,
  shadows,
  decorations.pathreplacing,
  decorations.markings,
  patterns
}

\graphicspath{{figs/}}

\newcommand{\defeq}{\mathrel{\triangleq}}

\newcommand{\Bb}{\mathbf{b}}

\newcommand{\ab}{\mathbf{a}}
\newcommand{\bb}{\mathbf{b}}

\newcommand{\vb}{\mathbf{v}}

\newcommand{\V}{\mathbf{v}}
\newcommand{\Sb}{\mathbf{s}}

\newcommand{\Q}{\mathbf{q}}
\newcommand{\R}{\mathbf{r}}

\newcommand{\Pb}{\mathbf{p}}

\newcommand{\ceil}[1]{\lceil{#1}\rceil}

\newcommand{\iid}{i.\@i.\@d.\ }

\newtheorem{lemma}{Lemma}

\newtheorem{theorem}[lemma]{Theorem}

\theoremstyle{definition}

\newtheorem{egdummy}{Example}
\newenvironment{example}[1][]
{
	\begin{egdummy}[#1]
		\upshape
		
	}
	{
		\qed
	\end{egdummy}
}

\newtheoremstyle{myremark}
{\topsep}{\topsep}{\normalfont}{\parindent}{\itshape}{:}{ }{}

\theoremstyle{myremark}

\newcounter{Remark}
\newenvironment{Remark}
{
	\refstepcounter{Remark}
	\textbf{Remark \theRemark:}
}

\newcounter{Algorithm}

\newcommand\shortintertext[1]{
	\ifvmode\else\\\@empty\fi
	\noalign{
		\penalty0
		\vbox{\mathstrut}
		\penalty10000
		\vskip-\baselineskip
		\penalty10000
		\vbox to 0pt{
			\normalbaselines
			\ifdim\linewidth=\columnwidth
			\else
			\parshape\@ne
			\@totalleftmargin\linewidth
			\fi
			\vss
			\noindent#1\par}
		\penalty10000
		\vskip-\baselineskip}
	\penalty10000}

\makeatletter
\def\endthebibliography{
	\def\@noitemerr{\@latex@warning{Empty `thebibliography' environment}}
	\endlist
}
\makeatother
\usepackage[utf8]{inputenc}

\title{Beyond Yao's Millionaires: Secure Multi-Party Computation of Non-Polynomial Functions}
\author[]{Seyed Reza Hoseini Najarkolaei}
\author[]{Mohammad Mahdi Mojahedian}
\author[]{Mohammad Reza Aref}

\affil[]{\footnotesize Information Systems and Security Lab. (ISSL), Sharif University of Technology, Tehran, Iran
\\hoseini.education@gmail.com, m.mojahedian@gmail.com, aref@sharif.edu}

\allowdisplaybreaks

\begin{document}

\maketitle

\begin{abstract}
In this paper, we present an unconditionally secure $N$-party comparison scheme based on Shamir secret sharing, utilizing the binary representation of private inputs to determine the $\max$ without disclosing any private inputs or intermediate results. Specifically, each party holds a private number and aims to ascertain the greatest number among the $N$ available private numbers without revealing its input, assuming that there are at most $T < \frac{N}{2}$ honest-but-curious parties. The proposed scheme demonstrates a lower computational complexity compared to existing schemes that can only compare two secret numbers at a time. To the best of our knowledge, our scheme is the only information-theoretically secure method for comparing $N$ private numbers without revealing either the private inputs or any intermediate results. We demonstrate that by modifying the proposed scheme, we can compute other well-known non-polynomial functions of the inputs, including the minimum, median, and rank. Additionally, in the proposed scheme, before the final reveal phase, each party possesses a share of the result, enabling the nodes to compute any polynomial function of the comparison result. We also explore various applications of the proposed comparison scheme, including federated learning.
\end{abstract}

\section{Introduction}
\label{section:introduction}
%N party computation secure infor theoretical. we will focus on max computation as an example. non polynomial function eval as well.
Given the substantial surge in data generated by diverse sensors and devices, processing and analyzing such vast amounts of data on a single machine has become unfeasible. To address this challenge, distributed computing techniques have been devised, involving the partitioning of computation tasks into smaller sub-tasks that are then outsourced to distributed processing nodes \cite{van2002distributed}. However, when computation tasks are delegated to distributed nodes, ensuring the privacy of data becomes a critical concern. Secure multi-party computation (MPC) techniques have been employed to execute computations across multiple distributed workers while safeguarding data privacy \cite{goldreich1998secure}. MPC enables a set of nodes to compute a joint function of their local data without compromising the confidentiality of their private inputs. Since its inception by Yao in \cite{yao1982protocols}, MPC has found applications in various domains such as secret sharing \cite{shamir1979share}, secret multiplication \cite{ben1988completeness}, large-scale matrix multiplication \cite{najarkolaei2020coded,nodehi2018entangled}, voting \cite{najarkolaei2022information}, and more.

A pivotal component in MPC is secure comparison, often referred to as the millionaires' problem. The objective is to compare two secret numbers without disclosing any information about them except for the ultimate result. Secure comparison finds applications in diverse problem domains, including secure online auctions \cite{damgard2008homomorphic}, machine learning and data mining \cite{bost2014machine}, and secure sorting algorithms \cite{dehghan2020secure}. For instance, consider the following scenarios:

\begin{enumerate}
	\item Assume that $N$ parties participate in an auction and the auctioneer wants to sell the item to the highest bidder, however, due to the privacy constraints, the parties are not willing to announce their bids except the winning one.
	
	\item Assume that a federal agency wants to sort the total wealth of different companies to publish it for the investors. However, these companies do not want to reveal the exact wealth to one another.
	
	\item In a federated learning system, different data owners (clients) compute their local gradients, and the central server wants to compute the coordinate-wise median, known for its robustness against malicious data owners, without revealing the actual gradients or model updates from individual clients.
	
	%\item Two people have the same position in a company. They want to find out that whether their salaries are equal to each other or not. But they are not willing to tell the salary to any one.
	
	%\item Consider $n_1$ students from school $S_1$, $n_2$ students from school $S_2$, ...., and $n_K$ students from school $S_K$. Assume that all of the students are interested in finding the school with the best education system. The comparison parameter is that in which of the schools the minimum GPA is higher than the others.
	
	%\item A company wants to hire ten college graduates students based on their GPA. How would the company find the top 10 candidates without gaining any information about the GPA of all of the candidates?
\end{enumerate}

These scenarios are few from a lot of examples where comparison, sorting, $\max$ and median computations are needed. In fact, the basis of all these is the multi-party comparison. To this end, in this work, we study how to compare two or more unknown inputs in an information-theoretically secure manner. 

Most of the schemes in the literature are computationally secure, i.e., if we have an adversary with high computation power or storage capacity (which may be possible in the near future with the developments in quantum computers), it can violate privacy. 
There are few unconditionally secure comparison protocols. In \cite{david2015efficient}, an unconditionally secure comparison protocol was proposed. This protocol employs a commodity-based model, with the requirement that a trusted initializer distributes correlated randomness (commodities) to the parties before the start of the protocol execution. In \cite{damgaard2006unconditionally}, the first bit-decomposition based unconditionally secure comparison is proposed, albeit with significant computational complexity. The approach presented in \cite{nishide2007multiparty} introduces some unconditionally secure operations, such as \emph{less than, bit decomposition}, and \emph{equality test}, with lower complexity compared to the method in \cite{damgaard2006unconditionally}. Another recently proposed unconditionally secure comparison, called Rabbit \cite{makri2021Rabbit}, exhibits lower complexity compared to the algorithms presented in \cite{damgaard2006unconditionally, nishide2007multiparty}. All of these schemes \cite{damgaard2006unconditionally, nishide2007multiparty, makri2021Rabbit} (utilized as benchmarks) can be applied for the comparison of just \emph{two secrets} shared among $N$ parties. Specifically, when comparing more than two numbers, one must iteratively apply these schemes for pairwise comparisons, potentially leading to information leakage about the private inputs. Moreover, these schemes demand significant computational power and communication capacity. Additionally, they cannot be employed for further computations on the final result without risking information leakage.

%Similar to the existing secure comparison literature, most of the online auction protocols are only computationally secure.

\subsection{Our Contributions}

Motivated by these, in this work, we propose an unconditionally-secure multi-party comparison and non-polynomial function evaluation scheme.
In this paper, we introduce partition and $0$-coded vector of the private inputs, which enable the other parties to compare the secret inputs with lower complexity compared to the conventional scheme and in a straightforward way, while satisfying strong information-theoretic privacy guarantee. We also propose a new method for equality test  and combine these schemes to be able to compute the $\max$ of more than two secrets and other non-linear functions. To illustrate our algorithm concisely, each party first computes its binary representation of its private number, then constitutes vectors of this binary representation called a partition vector and a $0$-vector. Using a theorem that we will prove, comparing two numbers is equivalent to the existence of only one zero in the difference of the partition representation of one number and the $0$-vector representation of another number. We then use Shamir's secret sharing to distribute partition and $0$-vector representations among $N$ parties in a way that nobody can gain any information about the private numbers. Our main contributions are listed as follows. Our main contributions are listed below.

\begin{enumerate}
	\item \textbf{Security:} The proposed scheme is unconditionally secure against any subset of $T< N/2$ colluding nodes even if these nodes have unbounded computation power and storage.

	\item \textbf{Novelty in coding the inputs:} In this paper, given a number $s$ in the finite field $\mathbb{F}_q$, we begin by constructing its binary representation. Subsequently, based on this binary representation, we introduce $0$-coded and partition vector representations. These representations facilitate the introduction of a novel coding method, enabling nodes to perform comparisons and compute non-polynomial functions, such as $\max$ and median, with low complexity while maintaining privacy.
	
	\item \textbf{Simplicity and Efficiency:} Our scheme can be implemented in a straightforward way. In addition, the complexity of the proposed scheme is much lower than the existing unconditionally secure comparison methods.
	
%	\item The proposed scheme is robust against node failure. If less than a threshold, $T$, of the nodes drops, the algorithm will not fail.
%	\item The proposed scheme is scalable. At each step of the protocol, if some new nodes join the protocol, the existing parties just send a function of their private input to the new nodes, without modifying the previous shares sent to the other existing parties.

	\item \textbf{Computing other useful functions:} In this paper, we study how to compare two or more numbers in an unconditionally secure manner. In addition to comparison, we show how to privately compute $\max$ function of more than two secrets, $\mathrm{median}$ of the secret inputs, equality test, $\min$ function, etc. 
	
	\item \textbf{Computing non-linear functions:} The proposed scheme is based on Shamir secret sharing \cite{shamir1979share}. Thus, we can use the proposed algorithm to compute non-polynomial functions of the secret inputs. For example, It is shown that median of gradient vectors is a good approximation of correct direction in federated learning. Also, gradient vectors which are far from the median can be detected as outliers. So, it is important to compute the distance of each gradient vectors from median ($|g_i - \mathrm{median}(\{g_k\}_{k=1}^N)|^2$, $\forall i \in [N]$) and remove some of them as outliers, where $\{g_k\}_{k=1}^N$ are gradient vectors received from the all of the nodes.
	One can see that we cannot compute these kind of functions with the conventional multi-party computation protocols. To the best of our knowledge, this scheme is the only information-theoretic secure solution that can compute non-polynomial functions of the private inputs.
	%\item We can extend this scheme to the case that some of the parties are malicious adversaries. The extended scheme is also robust against at most $T$ node failure.
\end{enumerate}

\subsection{Related Works}

The secure comparison problem was first studied by Yao in \cite{yao1982protocols} by using cryptographic tools. Yao's solution is exponential in time and requires large memory which makes it impractical. After Yao's initial attempt to solve the secure comparison problem, various protocols have been proposed to reduce the computation cost, required memory, and communication load \cite{beaver1989multiparty,chaum1987multiparty,grigoriev2014yao,ioannidis2003efficient,blake2004strong,lin2005efficient,liu2017efficient,abspoel2019fast,yang2009efficient,chong2019circular,liu2014multi,li2014efficient,jiang2020semi,damgaard2006unconditionally,nishide2007multiparty,garay2007practical,shundong2008symmetric,damgaard2007efficient,damgard2008homomorphic,damgaard2008correction,kolesnikov2009improved,veugen2022lightweight,damle2019practical,david2015efficient,makri2021Rabbit,lipmaa2013secure,toft2011sub,yu2012probabilistically}. These protocols employ various approaches to address the secure comparison problem, which we briefly summarize in the following. %References \cite{beaver1989multiparty} and \cite{chaum1987multiparty} use multi-party circuit computation to compare the secret numbers. Grigoriev et al. proposed a comparison protocol in \cite{grigoriev2014yao} based on some classical physics laws. Their method cannot be applied to the online setting due to the time consuming nature of their protocol.

\subsubsection{Garbled Circuit}
Garbled circuit is a cryptographic tool that enables two parties to compute a function of their private inputs in collaboration with each other without utilizing any trusted third parties and revealing their private data \cite{yao1986generate, micali1987play}. Particularly, garbled circuits enable secure comparison when the desired function can be described as a Boolean circuit, e.g., comparison, summation. Many works utilize garbled circuits in performing secure comparison \cite{kolesnikov2009improved, huang2011faster,huang2011efficient,zahur2015two}.

\subsubsection{Quantum-Based Comparison}

The first quantum private comparison (QPC) protocol was proposed in \cite{yang2009efficient}. Many works including \cite{chong2019circular,liu2014multi,li2014efficient} utilize QPC protocols. One can see that these quantum-based protocols require all nodes to have full quantum capability and a trusted third party. Reference \cite{jiang2020semi} proposed a semi-quantum private comparison protocol based on Bell states and a trusted semi-honest quantum third party. In this setting, authors in \cite{jiang2020semi} assume that the third party has quantum capability and not collude with other participants.

\subsubsection{Homomorphic Encryption}

Many works have implemented secure comparison using homomorphic encryption \cite{damgaard2007efficient,damgard2008homomorphic,damgaard2008correction,blake2004strong,lin2005efficient,liu2017efficient,kolesnikov2009improved,damle2019practical}. 
The proposed schemes in \cite{blake2004strong,lin2005efficient} are single-round solutions that use Paillier homomorphic encryption and zero-knowledge proof to preserve the privacy. %it linearly depends on the length of the secrets. 
Recently, another protocol based on Paillier encryption was proposed in \cite{liu2017efficient}, which uses vectorization of the private input.
\cite{veugen2022lightweight} introduced a new secure comparison protocol that uses homomorphic encryption. This protocol is lightweight, and needs lower computation power compared to \cite{damgaard2007efficient,damgard2008homomorphic,damgaard2008correction,kolesnikov2009improved}. The advantage of the scheme in \cite{veugen2022lightweight} is that it does not need intermediate decryption, which is computationally expensive.
%In \cite{damle2019practical}, a collusion resistance with the complexity of $O(1)$ per comparison was proposed that uses Pedersen commitment, but the protocol needs a central server.

\subsubsection{Arithmetic Black-Box}
Another set of comparison schemes are built on top of arithmetic black-box (ABB) \cite{cramer2001multiparty}. The arithmetic black-box enables some parties to securely store and reveal secrets, and perform arithmetic operations on these secrets. From a functionality point of view, ABB can be thought of as a trusted third party, which stores elements of the field as well as performs arithmetic computations on them. 
ABB can be implemented by using varying tools such as oblivious transfer \cite{even1985randomized} or homomorphic encryption \cite{paillier1999public}. There are many protocols in this category including \cite{lipmaa2013secure,toft2011sub,yu2012probabilistically} and they vary greatly in structure.

\subsubsection{Bit-Wise Operations}
Another line of secure comparison protocols are based on bit-wise operations and bit decomposition \cite{damgaard2006unconditionally,nishide2007multiparty,garay2007practical,ioannidis2003efficient,makri2021Rabbit}. 
A two-round protocol that uses complex bit-wise operations was proposed in \cite{ioannidis2003efficient}. The proposed scheme in \cite{ioannidis2003efficient} is polynomial in time and communication.
The schemes in \cite{damgaard2006unconditionally,nishide2007multiparty,garay2007practical} are based on bit decomposition of the secret and apply some Boolean operations.
Computing bit decomposition of a secret is computationally expensive. Recently, a new method called Rabbit is proposed \cite{makri2021Rabbit}. The Rabbit scheme is based on bit-wise operation without computing bit decomposition of the secret, which makes it faster than conventional schemes. However, this scheme also uses a third party as an ABB for two secrets multiplication.
%In \cite{shundong2008symmetric}, a protocol based on symmetric cryptography was proposed. 
We note that all of these proposed protocols in \cite{damgaard2006unconditionally,nishide2007multiparty,garay2007practical,ioannidis2003efficient,makri2021Rabbit} are for two-secret comparison and they cannot be extended to the computation of $\max$ function of more than two secrets or computing a non-polynomial function of the secrets, such as median, unlike our proposed secure comparison technique.

It is worth noting that this paper aligns with category 5 of secure comparison algorithms. Within this category, we compared our algorithm with those introduced in \cite{damgaard2006unconditionally,nishide2007multiparty,makri2021Rabbit} in terms of complexity.

The rest of the paper is organized as follows: In Section \ref{section:problemsetting}, we introduce the problem setting and describe some technical preliminaries that we use. Section \ref{section:proposedcomparison} presents the proposed secure comparison scheme along with its application to $\max$ computation. In Section \ref{section:application}, we discuss further applications of the proposed secure comparison scheme. Finally, Section \ref{section:conclusion} concludes the paper.

\textbf{Notation.} In this work, we represent vectors with lowercase bold letters. The element-wise multiplication of two vectors $\ab$ and $\bb$ is denoted by $\ab\odot\bb$. We use $[N]$ to represent the set $\{1,2,\dots,N\}$. The transpose of a vector $\vb$ is represented by $\vb^\intercal$. Sets are shown using calligraphic font.

%%%%%%%%%%%%%%%%%%%%%%%%%%%%%%%%%%%%%%%%%%%%%%%%%%%%%%%%%%%%%%%%%%%%%%%%%%%

\section{System Model and Preliminaries}
\label{section:problemsetting}

%\subsection{System Model}
%\label{subsection:systemmodel}

We consider a decentralized secure comparison system with $N$ nodes. Each node is connected to every other other node by a point-to-point private link. Each party has a private number $s^{(i)}$, $i \in [N]$, from a field $\mathbb{F}_q$, where $q$ is a prime number. Each node in the system wants to find the greatest number among these $N$ private numbers without revealing its private input. That is, our goal is to compute the maximum among these $N$ numbers that are distributively held by $N$ nodes without leaking any information beyond the final result. In our model, up to $T$ of the nodes are semi-honest, where $T < \frac{N}{2}$. Semi-honest nodes follow the protocol but may collude with each other to gain some information about the private inputs of the others. 
% info theoretically secure, side benefit function evaluation

To securely compute the maximum number, in the proposed scheme, each node $i$ sends a function of its private number, shown by $F_j^{(i)}\left(s^{(i)}\right)$, to each other node $j$, where $F_j^{(i)}: \mathbb{F}_q \rightarrow \mathbb{F}_q$. Then, nodes perform certain computation tasks in collaboration with each other. Finally, each node sends a share of the final result to all of the other nodes to allow them to compute the final result of the comparison. Before we present the proposed scheme, we give an overview certain preliminaries that we utilize in our algorithm.

%\section{Preliminaries}
%\label{section:preliminaries}

\subsection{Shamir Secret Sharing}
\label{subsection:shamirsecretsharing}

Assume that a node, called \emph{dealer}, wants to share a secret among some nodes such that any subset of at least $T$ nodes can recover the actual secret $s$, while any subset of fewer than $T$ nodes cannot gain any information about the secret. The first secret sharing scheme was proposed in 1979 by Shamir \cite{shamir1979share} and Blakley \cite{blakley1979safeguarding} independently. Key safeguarding was the initial motivation of secret sharing, but now it is a basic cryptographic tool that is used in many application such as e-voting, crypto-currencies, etc. 

In Shamir secret sharing (SSS), the dealer constructs a polynomial $p(x)=s+r_1x+r_2x^2+\dots+r_{T-1}x^{T-1}$ of degree $T-1$, where the constant term is the secret $s$, and the other coefficients are chosen uniformly and independently at random from the field $\mathbb{F}_q$. Then, the dealer sends a distinct point $p(x)$ to each of the nodes. By using Lagrange interpolation, any subset of $T$ nodes can recover the polynomial, but any subset of less than $T$ nodes cannot understand anything about the secret $s$. %Nowadays, Shamir secret sharing is widely used in different coded computing and multi-party applications, e.g., martix multiplication \cite{nodehi2018entangled,hoseini2020coded}, voting \cite{hoseini2022votingjournal,binu2016secret}, secure aggregation \cite{bonawitz2017practical,bonawitz2019federated,so2021lightsecagg,jahani2022swiftagg,jahani2022swiftagg+}, etc.
In some cases, the dealer may be malicious and send some random number instead of sending the points that are located on the constructed polynomial.
In 1985, Chor et al.\cite{chor1985verifiable} extended SSS to verifiable secret sharing, which enables the nodes to verify whether their shares are consistent. In other words, in verifiable SSS, the nodes verify that their shares are indeed on polynomial of degree $T-1$.

\subsection{Random Secret Generation}
\label{subsection:randomsecretgeneration}
In many applications, nodes in the system need to agree on a random and unknown secret $s$ \cite{bar1989non}. In random secret generation, first, distinct and non-zero $\alpha_{1},\alpha_{2},\dots,\alpha_{N}$ are chosen uniformly and independently at random from field $\mathbb{F}_q$, which are known by all of the nodes and $\alpha_i$ is assigned to node $i$, $\forall i \in [N]$. Then, each node $i$ generates a random secret $s_i$ and constructs a random polynomial $p_i(x)=s_i+r^{(i)}_1x+r^{(i)}_2x^2+\dots+r^{(i)}_{T}x^{T}$ and sends $p_i(\alpha_j)$ to each node $j$, $\forall i,j \in [N]$. Let us define $p(x)\defeq \sum_{j=1}^Np_j(x)$. Each node $i$ has $p_1(\alpha_i),p_2(\alpha_i),\dots,p_N(\alpha_i)$, and it computes $p(\alpha_i)= \sum_{j=1}^Np_j(\alpha_i)$. Thus, each node has a share of $p(0)$ which is equal to $ \sum_{j=1}^N s_j$. That is, at the end of random secret generation, each node has a share of the overall secret but none of the nodes actually knows the secret value $s$. Random secret generation is private and the its complexity can be upper bounded by 1 multiplication invocation \cite{damgaard2006unconditionally}. %The algorithm of random secret generation is presented in Algorithm \ref{algorithm:randomsecretgeneration}.

%\begin{algorithm}
%\label{algorithm:randomsecretgeneration}

%\textit{Initialization:} A field $\mathbb{F}_q$ and $N$ non-zero and distinct elements $\alpha_1,\alpha_2,\dots,\alpha_{N} \in \mathbb{F}_q$ are known to all of the nodes, where $q$ is a prime number.

%\textit{The protocol:} 
%	\begin{enumerate}
%		\item Each node $n$ constructs a random polynomial $p_n(x)=s_n+r^{(n)}_1x+r^{(n)}_2x^2+\dots+r^{(n)}_{T}x^{T}$ and sends $p_n(\alpha_{n'})$ to node $n'$ for all $n,n' \in [N]$, by using SSS. 

%		\item
%		Each node $n$ computes $\displaystyle\sum_{j=1}^Np_j(\alpha_n)$, and store it as a share of random and unknown secret $\displaystyle\sum_{j=1}^Np_j(0)$.
		
%	\end{enumerate}
	
%\end{algorithm}

\subsection{Partition and $0$-Coded Vectors}
\label{subsection:partitionand0codedvector}

Let $s$ be a number in the field $\mathbb{F}_q$ and $\overline{s_1s_2\dots s_L}$ be its binary representation of length $L$. We define \textit{partition vector} $\V_{\!s}$ of $s$ as follows

\begin{align}\label{partition}
    \V_{\!s} &= \begin{bmatrix}
           \overline{s_1} \\
           \overline{s_1s_2} \\
           \vdots \\
           \overline{s_1s_2\dots s_L}.
         \end{bmatrix}
\end{align}

Next, we define $\V_{s}^{0}$ as the \textit{$0$-coded vector} of $s$ which is equal to 

\begin{align}\label{zero-coded}
\V_s^0=\begin{bmatrix}
           z_1 \\
           z_2 \\
           \vdots \\
           z_L
         \end{bmatrix}
         \text{, where $z_i=$}
	\begin{cases}
		\overline{s_1s_2\dots s_{i-1}1} &\text{ if $s_i=0$}, \\
		r_i &\text{ otherwise},\\
	\end{cases}
\end{align}
such that $r_i$ is a random binary number of length unequal to $i$, $\forall i \in [L]$. Here, our definition of a $0$-coded vector is a vectorized version of the $0$-encoding set introduced in \cite{lin2005efficient}, where the $0$-encoding and $1$-encoding sets of a string $s=s_1s_{2}\dots s_L \in \{0,1\}^L$ are defined as follows:
\begin{align}
\mathcal{S}_s^0&=\{s_1s_{2}\dots s_{i-1}1| s_i=0, 1 \leq i\leq L\}\\
\mathcal{S}_s^1&=\{s_1s_{2}\dots s_{i}| s_i=1, 1 \leq i\leq L\}
\end{align}
In \cite{lin2005efficient}, it is demonstrated that $a>b$ if and only if $\mathcal{S}_a^1$ and $\mathcal{S}_b^0$ have a common element.

\begin{theorem}
\label{theorem1}
Assume that $a$ and $b$ are two numbers with length-$L$ binary representations of $\overline{a_1a_2\dots a_L}$ and $\overline{b_1b_2\dots b_L}$, respectively. We have $a>b$ if and only if vector $\V_a-\V_b^0$ has exactly one $0$ entity.
\end{theorem}
\begin{proof}
Let us assume that $a$ is greater than $b$. Thus, there exists a position $i \in [L]$, where $a_i=1$ and $b_i=0$, with $a_j=b_j$ for all $j<i$. The $i$th element of the partition vector of $a$ is equal to the $i$th element of 0-coded vector of $b$, which is $\overline{a_1a_2...a_{i-1}1}$. Thus, the $i$th entity of $\V_a-\V_b^0$ is equal to $0$.

On the other hand, let us assume that the $i$th entity of $\V_a-\V_b^0$ is equal to $0$. In this case, one can see that $b_i=0$, otherwise, it would not possible from (\ref{zero-coded}) that the $i$th element of $\V_b^0$ is equal to $\V_a$.
Hence, $a_1a_2...a_i=b_1b_2...b_{i-1}1$. Thus, $a>b$, while $a_i=1$ and $b_i=0$, and $a_j=b_j$ for all $j<i$. 
\end{proof}

\begin{example}
For ease of understanding, in this part we illustrate Theorem \ref{theorem1} with an example. Assume that $a=10$ and $b=9$. Binary representation of $a$ and $b$ are $1010$ and $1001$. One can see that $\V_a=[1,10,101,1010]^T$ and $\V_b^0=[11, 11,101,100]^T$. Hence, $\V_a-\V_b^0=[1110, 1111, 0 , 10]$ and the third entity is equal to $0$, so we can conclude that $a$ is greater than $b$. 
\end{example}

%Keeping Theorem~\ref{theorem1} in mind, we now formulate the secure comparison problem.
%\textbf{Notation.} In this work, we represent vectors by uppercase bold letters. $\ab*\Bb$ denotes the element-wise multiplication of two vectors $\ab$ and $\Bb$. We let $[N]$ denote the set $\{1,2,\dots,N\}$. We represent the transpose of a vector $\V$ by $\V^T$. Finally, we show sets with calligraphic font.

\section{Proposed Unconditionally Secure Comparison Scheme}
\label{section:proposedcomparison}

In this section, we describe the proposed scheme for distributed $N$-party secure maximum computation. The proposed scheme does not sacrifice the privacy of the private inputs (secrets) and does not reveal any intermediate values to the participating parties but the final result. To this end, we first propose a secure comparison scheme, where all $N$ nodes collaboratively compare two secret numbers. This secure comparison scheme is a building block in our $N$-number $\max$ computation system. Next, by modifying the proposed secure two-number comparison system, we present a secure comparison indicator (SCI) that outputs $0$ if the first number is greater than the second one, and outputs $1$ otherwise. Another building block of our $N$-number comparison system is the Secure Comparison Gate (SCG) that receives the partition and $0$-coded vectors of secret numbers $a$ and $b$ as inputs, and outputs the partition and $0$-coded vector of the $\max(a,b)$. SCG is critical because it allows us to carry on next comparison without revealing the intermediate maximums. Finally, we design a circuit based on SCGs to find the maximum of $N$ secrets in a distributed manner.  

%This section includes five subsections. In Subsection , we propose a secure comparison scheme, where  We propose a scheme for equality test in Subsection \ref{subsection:equality and zero indicator}.
%In Subsection \ref{subsection:securecomparisonindicator}, we propose a  protocol that outputs $0$ or $1$. The output of the proposed scheme is $0$, , and it is  in all of the other cases.
 
As noted earlier, in our $N$-party comparison system, we assume that at most $T$ of the nodes are semi-honest and may collude with each other to obtain some information about the private inputs of the other nodes.

\begin{Remark}
	One of the main blocks in our scheme is secure two-secret multiplication, which requires $N>2T$ \cite{ben1988completeness}. Thus, in the proposed algorithm, we assume that $N$ is greater than $2T$.
\end{Remark}

\begin{Remark}
	\label{remark:multiplicationinvocation}
Since the invocation of multiplication of two secrets is the dominant complexity of our proposed secure comparison technique, similar to \cite{damgaard2006unconditionally,nishide2007multiparty}, we define the complexity of a comparison scheme as the number of invocations of secrets multiplication protocol.
\end{Remark}

\begin{Remark}
	In our scheme, we assume that distinct and non-zero $\alpha_{1},\alpha_{2},\dots,\alpha_{N}$ are chosen uniformly and independently at random from field $\mathbb{F}_q$ in advance, which are known by all the nodes and $\alpha_j$ is assigned to node $j$, $\forall j \in [N]$.
\end{Remark}

\subsection{Secure Comparison}
\label{subsection:securecomparison1}

In this section, we describe how to compare two secrets in an information-theoretically secure manner. This is a fundamental algorithm in $\max$ computing of $N$ secret numbers.
 
We have $N$ nodes and two of them, denoted by $A_1$ and $A_2$, have secrets $s_1$ and $s_2$, respectively. All of the nodes are interested in computing the maximum of $s_1=\overline{s^{(1)}_1s^{(1)}_2\dots s^{(1)}_L}$ and $s_2=\overline{s^{(2)}_1s^{(2)}_2\dots s^{(2)}_L}$.
The proposed algorithm is as follows.

\textbf{Sharing Phase.}
In the sharing phase, $A_1$ and $A_2$ compute $\V_{s_1}$ and $\V_{s_2}^0$, respectively, and share them by using Shamir secret sharing \cite{shamir1979share}. More precisely, $A_1$ constructs a polynomial $\Pb_1(x)=\V_{s_1}+\R^{(1)}_1x+\R^{(1)}_2x^2+\dots+\R^{(1)}_Tx^T$, where $\R^{(1)}_i$ are chosen uniformly and independently at random from the field $\mathbb{F}^{L}$, and sends $\Pb_1(\alpha_j)$ to each node $j$, $\forall j \in [N]$. Similarly,  $A_2$ constructs a polynomial $\Pb_2(x)=\V^0_{s_2}+\R^{(2)}_1x+\R^{(2)}_2x^2+\dots+\R^{(2)}_Tx^T$, where $\R^{(2)}_i$ are chosen uniformly and independently at random from the field $\mathbb{F}^{L}$, and sends $\Pb_2(\alpha_j)$ to each node $j$, $\forall j \in [N]$.

\textbf{Random Secret Generation Phase.} In this phase, all of the nodes execute random secret generation in collaboration with each other as described in Subsection \ref{subsection:randomsecretgeneration} over the field $\mathbb{F}_q$. At the end of this phase, each node $i$ obtains the value of a degree-$T$ polynomial $p(x)$ at $x=\alpha_i$, where $p(0)$ is the secret value generated by the nodes in collaboration. We emphasize that none of the nodes actually knows this secret number.

\textbf{Vector Computation Phase.}
In this phase, first, each node $i$ computes $\Pb_1(\alpha_i)-\Pb_2(\alpha_i)$. Let us define $\Q(x) \defeq \Pb_1(x)-\Pb_2(x)$. One can see that $\Q(0)$ is equal to $\Pb_1(0)-\Pb_2(0)= \V_{s_1}-\V^0_{s_2}$, and each node $i$ has $\Q(\alpha_i)$. Hence, each node $i$ has a share of secret $\Sb \defeq \V_{s_1}-\V^0_{s_2}$. From Theorem \ref{theorem1}, if there exists a $0$ entity in vector $\Sb =[s'_1,s'_2,\dots,s'_L]^T$, then we can conclude that $s_1$ is greater than $s_2$.

At this stage, each node $i$ could reveal $\Q(\alpha_i)$ (sends $\Q(\alpha_i)$ to all other nodes), and then each of the nodes would be able to compute $\Q(x)$, and accordingly derive $\Sb$. With this, each of the nodes would be able to verify whether there is a $0$ entity in the vector $\Sb$. However, this approach leaks additional information about the secrets $s_1$ and $s_2$ beyond the final comparison result. In order to avoid such leakage and preserve the privacy of the secret numbers before revealing the final result, the nodes perform the next phase.

\textbf{Entity Computation Phase.}
Remember that each node $i$ has $\Q(x)=[q_1(x),q_2(x),\dots,q_L(x)]^T$ at $x=\alpha_i$ from the previous phase along with $p(\alpha_i)$ that was derived in the random secret generation phase. We note that $\Q(0)$ is equal to $\Sb \in \mathbb{F}_q^L$, while $p(0)$ is a random secret in $\mathbb{F}_q$. One can observe that $\prod_{j=1}^{L}q_j(0)$ is equal to $0$ if and only if there is a $0$ entry in the vector $\Sb = [s'_1, s'_2, \dots, s'L]^T$, indicating that $s_1 > s_2$. If $\prod{j=1}^{L}q_j(0)$ is not equal to $0$ and the nodes obtain its exact value, it may disclose information about $q_1(0), q_2(0), \dots, q_L(0)$. For instance, consider $L=3$, and if the nodes comprehend that the product of $q_1(0), q_2(0), q_3(0)$ equals $2$, they will deduce that the vector $[q_1(0), q_2(0), q_3(0)]$ cannot be equal to $[1, 1, 3]$. To safeguard privacy, instead of computing $\prod_{j=1}^{L}q_j(0)$ at the nodes, they calculate $p(0)\left(\prod_{j=1}^{L}q_j(0)\right)$ using secure MPC \cite{ben1988completeness, nodehi2019secure}. Consequently, each node $i$ obtains $s(\alpha_i)$, where $s(x) = p(x)\left(\prod_{j=1}^Lq_j(x)\right)$. By masking the exact value of $\prod_{j=1}^{L}q_j(0)$ with an unknown random number $p(0)$, the privacy of each element in $[q_1(0), q_2(0), \dots, q_L(0)]$ is preserved.

\textbf{Reconstruction Phase.} Each node $i$ gains $s(\alpha_i)$ from the previous phase. In this phase, each node $i$ broadcasts $s(\alpha_i)$ to other nodes, and therefore, by using Lagrange interpolation each of the nodes can compute $s(x)$ and derive $s(0)$ which is equal to $p(0)\left(\prod_{j=1}^{L}q_j(0)\right)$. We note that if $\Sb$ has a zero entity (meaning if secret $s_1$ is greater than $s_2$), we have $s(0)=0$, otherwise we have $s(0) \neq 0$ with high probability (since the randomly generated secret $p(0)$ is not equal to $0$ with high probability, see Remark~\ref{remark_randomgen} for further discussion). Thus, at the end of the reconstruction phase, each of the nodes deduces the maximum among the secret numbers $s_1$ and $s_2$ without gaining any information about the other number.

As for the privacy, we note that all sub-protocols including secure multiplications and the random generation are unconditionally private. The only phase that some information can be leaked is the reconstruction phase, where the nodes reveal $s(x)$ to derive $s(0)$ which is equal to $p(0)\left(\prod_{j=1}^{L}q_j(0)\right)$. One can see that $\prod_{j=1}^{L}q_j(0)$ is masked with an unknown random number $p(0)$. Thus, the nodes cannot understand additional information beyond the comparison result.

\emph{\textbf{Complexity:}} 
As mentioned in Remark \ref{remark:multiplicationinvocation}, the computational complexity of the algorithm is determined based on the number of multiplication invocations. In this algorithm, the complexity of random secret generation can be upper-bounded by just one multiplication invocation, as discussed in Subsection \ref{subsection:randomsecretgeneration}. Additionally, the computation of $p(0)\left(\prod_{j=1}^Lq_j(0)\right)$ requires $L$ multiplication invocations. Moreover, the rest of the algorithm has a computational complexity of less than one multiplication invocation. Thus, the computation complexity of the proposed secure comparison algorithm is $L+2$, while the complexity of the comparison algorithm of Rabbit \cite{makri2021Rabbit} is approximately equal to $53L$. Further, the complexity of the comparison schemes in \cite{damgaard2006unconditionally} and \cite{nishide2007multiparty} are $188L \log L + 205L$ and $279L+5$, respectively. The complexity of the secure comparison protocols is compared in Table \ref{table:comparison}. %One of the differences of our algorithm with \cite{makri2021Rabbit,damgaard2006unconditionally,nishide2007multiparty} is that, in our algorithm, we output a random number when $a\leq b$, while in the latter they output $1$ in this case. 
	%Also, our algorithm receives the coded version of the secret, while in \cite{damgaard2006unconditionally,nishide2007multiparty,makri2021Rabbit}, they receive the simple version of the numbers as the secrets. 
	From Table \ref{table:comparison}, we see that our secure comparison algorithm has more than $50\times$ lower complexity compared to the existing algorithms in the literature.

\begin{table}[h!]
	\centering
	\begin{tabular}{||c c||} 
		\hline
		Protocol &  Complexity \\ [0.5ex] 
		\hline\hline
		\cite{damgaard2006unconditionally} &  $188L \log L + 205L$ \\ 
		\cite{nishide2007multiparty} &  $279L+5$ \\
		\cite{makri2021Rabbit} & $\approx 53 L$ \\
		Proposed   &  $L+2$ \\ 
 [1ex] 
		\hline
	\end{tabular}\caption{Comparing the complexity of secure two-number comparison algorithms among $N$ parties, based on the number of multiplication invocations in each algorithm.}
    %\vspace{10mm}
\label{table:comparison}
\end{table}

\begin{Remark}
All references \cite{damgaard2006unconditionally,nishide2007multiparty,makri2021Rabbit} typically compare $s_1$ and $s_2$ by sharing their exact values among the nodes. However, this paper introduces a new method where sources distribute encoded versions of the secrets, rather than their exact values, to reduce the computational burden on the nodes. This approach also enhances our ability to perform additional non-linear operations, such as calculating medians, sorting, and more.
\end{Remark}

\begin{Remark}\label{remark_randomgen}
In the reconstruction phase, if $p(0)=0$ the term $p(0)\left(\prod_{j=1}^{L}q_j(0)\right)$ becomes $0$, even if the term $\prod_{j=1}^{L}q_j(0)$ is not equal to $0$, in which case the nodes may decide on an incorrect result. However, this is a highly unlikely scenario since the probability of
$p(0)=0$ is $\frac{1}{q}$, where $q$ is size of the field. So, by selecting a large enough field, the probability of this event can be made almost zero. 

Further, upon agreeing on a randomly generated secret in the random secret generation phase, before computing $p(0)\left(\prod_{j=1}^{L}q_j(0)\right)$, the nodes can verify that $p(0)$ is not equal to $0$ by using the scheme in \ref{subsection:equality and zero indicator}. In the highly unlikely scenario that $p(0)$ is equal to $0$, the nodes execute the random secret generation algorithm again to create a new polynomial $p(x)$ such that $p(0) \neq 0$. Upon verifying that $p(0)\neq 0$, they compute $s(0)=p(0)\left(\prod_{j=1}^{L}q_j(0)\right)$. With this, we are sure that the comparison result is correct with probability $1$.
\end{Remark} 

\subsection{Equality Test and Zero Indicator}
\label{subsection:equality and zero indicator}
As mentioned in Subsection \ref{subsection:securecomparison1}, the output of secure comparison algorithm is $0$ if $s_1>s_2$, and the output is a random number in field $\mathbb{F}_q$ in other cases. In order to use the secure comparison algorithm in a circuit for comparing more than two numbers, we need to show the output of the comparison with a bit. To this end, in this subsection, we implement a $\mathrm{Zero}$ function that gets a number $a \in \mathbb{F}_q$ as an input, and outputs a bit where,

\begin{align}
	\mathrm{Zero}(a)=\begin{cases}
		0 &\text{ if $a=0$} \\
		1 &\text{ if $a \neq 0$}\\
	\end{cases}.
\end{align}

In other words, suppose a secret $a$ is shared among the nodes using a polynomial $r(x)$, where $r(0) = a$. All nodes collaborate privately to determine $\text{Zero}(a)$. Once implemented, this function can be added to the end of the secure comparison algorithm to show the output with a bit, where the output is $0$ if $s_1>s_2$, and it is $1$ in other cases.

Before we describe how to implement this $\mathrm{Zero}$ function, we would like to note that the $\mathrm{Zero}$ function can also be utilized to infer whether two secret numbers are equal. That is, assume that two unknown and private numbers $a$ and $b$ from the field $\mathbb{F}_q$ are shared among the parties and they want to compute the shares of $\mathrm{Equal}(a,b)$, where 

\begin{align}
	\mathrm{Equal}(a,b)=\begin{cases}
		0 &\text{ if $a=b$} \\
		1 &\text{ otherwise}\\
	\end{cases}.
\end{align}

In this case, one can see that $\mathrm{Equal}(a,b)=\mathrm{Equal}(a-b,0)=\mathrm{Zero}(a-b)$. Thus, the $\mathrm{Zero}$ function can be also used for equality test in a standalone manner.

In the zero indicator algorithm, it is assumed that the shares of $a$ is distributed among the nodes and the nodes want to obtain shares of $0$ if $a=0$, and shares of $1$ otherwise. For this, we utilize Fermat's little theorem, which is presented for completeness in Theorem~\ref{Fermat}.

\begin{theorem}\label{Fermat}
	Assume that $q$ is a prime number and $a$ is a non-zero number, where $q\nmid a$. Then,
	\begin{align}
		a^{q-1}\overset{q}{\equiv} 1.
	\end{align}
\end{theorem}

By using Fermat's little theorem, one can see that $a^{q-1}$ is equal to $0$ if and only if $a=0$, and it is equal to $1$ in all the other cases.

In a standalone implementation of the zero indicator, there is an unknown number $a$, which is already secret shared among the nodes. That is, there exists a polynomial $r(x)$ with constant term that equals to $a$, where $r(\alpha_i)$ is sent to node $i$, $\forall i \in [N]$.
By using secure multiplication \cite{ben1988completeness}, the nodes can compute the shares of $a^{q-1}$. In other words, they can construct a polynomial $g(x)$ privately in collaboration with each other, such that $g(0)$ is equal to $a^{q-1}$. Note that at first, each node $i$ has $r(\alpha_i)$, so they can continue the computation in a secure way until each node $i$ obtains $g(\alpha_i)$. Finally, 
if the computation is completed (the node do not need to do any other computation on top of this operation), 
the nodes reveal their corresponding points on $g(x)$ and thus, each of them can compute $g(x)$ and verify whether $g(0)$ is equal to $0$ or $1$, else, they use $g(x)$ to continue the computation without revealing anything.

As for the privacy, we note that the only sub-protocol used here is secure multiplication, which is an unconditionally private protocol as shown in \cite{ben1988completeness}. In \cite{fullproof}, it is shown that BGW multiplication scheme is unconditionally secure and broadcasting the result does not reveal anything about the private inputs beyond the output.
%The only step that some information is broadcasted to the other nodes is in the revealing of the multiplications which is equal to $g(x)$. While the multiplication is secure, so the nodes cannot understand anything about the input more than the result. 

\begin{Remark}
\label{remark:zeroatothepowerofq}Assume that $\lceil \log_2(q) \rceil = L$, and the binary representation of $(q-1)2$ is $q_1q_2\dots q_L \in {0,1}^{L}$. Let us define $a^{(0)}=a$. To compute $a^{q-1}$, all nodes initially collaborate to compute $a^{(1)} \equiv a^2$, followed by $a^{(2)} \equiv a^4$, and so on, up to $a^{(L-1)} \equiv a^{2^{L-1}}$ in $L-1$ rounds. Subsequently, $a^{q-1}$ can be expressed as $\prod_{j=1,q_{j}\neq0}^{L}a^{(L-j)}$. Therefore, the computation of $a^{q-1}$ requires at most $2L-2$ multiplication invocations.
\end{Remark}

\emph{\textbf{Complexity:}} 
As it mentioned in Remark \ref{remark:zeroatothepowerofq}, to compute $a^{q-1}$, the nodes perform secure multiplication for less than $2L=2\ceil{\log_2{(q)}}$ times. Since the length of binary representation of elements in $\mathbb{F}_q$, denoted by $L$, is $\ceil{\log(q)}$, the complexity of the proposed equality scheme is upper bounded by $2L$ whereas the complexity of the equality test proposed in \cite{nishide2007multiparty} is almost $81L$. We present a comparison of these complexities in Table~\ref{table:equality}, from which we see that the proposed equality test algorithm has more than $40\times$ lower computation complexity than the existing approaches.

	\begin{table}[h!]
		
		\centering
		\begin{tabular}{||c c||} 
			\hline
			Protocol &  Complexity \\ [0.5ex] 
			\hline\hline
			\cite{damgaard2006unconditionally} &  $94L \log L + 92$ \\ 
			\cite{nishide2007multiparty} &  $81L$ \\
			Proposed   &  $2L$ \\ 
			[1ex] 
			\hline
		\end{tabular}\caption{Complexity comparison of the equality test algorithms.}
		
		\label{table:equality}
	\end{table}

\subsection{Secure Comparison Indicator}
\label{subsection:securecomparisonindicator}

In this subsection, we combine the proposed secure comparison scheme and the equality test, i.e., zero indicator, to design a component named \emph{secure comparison indicator (SCI),} that compares two numbers from the field $\mathbb{F}_q$ and outputs $0$ if the first one is greater than the second one, and outputs $1$ in the other cases. SCI is a non-linear gate with two input vectors each of length $L$ and a single output which can be $0$ or $1$ based on the inputs. Precisely, 
\begin{align}
\mathrm{SCI}\left(a,b\right)=\begin{cases}
		0, &\text{ if $a>b$} \\
		1, &\text{ otherwise}\\
	\end{cases}.
\end{align}

We note that, in in Section \ref{subsection:securecomparison1}, we propose an algorithm to compare two secret numbers such that if $a>b$, the output is $0$, and when $b \geq a$, the output is a random number. Further, in Section \ref{subsection:equality and zero indicator}, we present the implementation of a zero indicator. One can see that we have
\begin{align}
\mathrm{SCI}\left(a,b\right)=\mathrm{Zero}\left(\mathrm{Comparison}(a,b)\right).
\end{align}
That is, the SCI gate can be implemented using the secure comparison and zero indicator algorithms as shown in Fig.~\ref{fig:sci}. We detail the implementation as follows:
\begin{enumerate}
	\item Because each node has a share of $\V_a$ and $\V_b^0$,  by running the comparison algorithm described in Subsection \ref{subsection:securecomparison1}, each nodes $i$ can derive $u(\alpha_i)$ in a private manner, where $u(0)=\prod_{j=1}^L q_j(0)$ and the degree of $u(x)$ is equal to $T$. It must be mentioned that in Subsection \ref{subsection:securecomparison1}, each node $i$ has access to $s(\alpha_i)$, where $s(0)=p(0)\big(\prod_{j=1}^L q_j(0)\big)$. But, in this part, because we will use zero indicator before revealing the final result, we do not need to multiply $\prod_{j=1}^L q_j(0)$ with a random secret $p(0)$.
 We note that $u(x)$ is constructed by using secure multi-party computation as detailed in \cite{ben1988completeness} and none of the nodes knows anything about the secret $u(0)$. As mentioned in Subsection \ref{subsection:securecomparison1}, if $a>b$, $u(0)$ is equal to $0$, else, $u(0)$ is a non-zero random number in $\mathbb{F}_q$. 
	
	\item As shown in Subsection \ref{subsection:equality and zero indicator}, the nodes then construct a polynomial $g(x)$, such that $g(0)$ is equal to $u(0)^{q-1}$. We note that each node $i$ has $u(\alpha_i)$, so they can continue the computation in a secure way until each node $i$ obtains $g(\alpha_i)$. Finally, using the Fermat's little theorem, the nodes can represent the result with a bit. 
\end{enumerate}

%\begin{figure}[htbp]
%	\centering
%	\includegraphics[draft=false,width=90mm]{SCI-1.PNG}
%	
%	\caption{\emph{SCI} is combination of secure comparison and and zero indicator. It takes partition vector of $a$ and $0$-coded vector of $b$, and outputs $0$ if $a>b$ or $1$ in other cases.}
%	\label{fig:sci}
%\end{figure}

\begin{figure}[htbp]
\centering
\resizebox {.6\columnwidth} {!} {
\begin{tikzpicture}[scale=1]
\tikzstyle{every node}=[font=\small]

%\draw[thick] (-.5,0) rectangle (2,1) node [pos=.5] {Encoder};

\draw[rounded corners,fill=black!10!white] (-0.5,-0.5) rectangle (5.5,2.5) node at (2.5,2) {SCI Module};

\draw[rounded corners,fill=black!30!white] (0,0) rectangle (2,1.5) node [align=center,pos=.5] {Secure\\Comparison};

\draw[rounded corners,fill=black!30!white] (3,0) rectangle (5,1.5) node [align=center,pos=.5] {Zero\\Indicator};

\draw[-latex] (2,0.75) -- (3,0.75);
\draw[-latex] (5,0.75) -- (6,0.75) node at (6.35,0.75) {$0/1$};

\draw[-latex] (-1,0.5) -- (0,0.5) node at (-1.35,0.5) {$\V_b^0$};
\draw[-latex] (-1,1) -- (0,1) node at (-1.35,1) {$\V_a$};

\end{tikzpicture}
}
\caption{The SCI Module is a combination of secure comparison and zero indicators. It takes the partition vector of $a$ and the $0$-coded vector of $b$, and outputs $0$ if $a>b$ or $1$ in other cases.}
\label{fig:sci}
\end{figure}
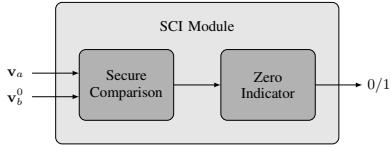

As for the privacy, SCI is a combination of two unconditionally secure protocols. Thus, the proposed SCI protocol is unconditionally secure.

\emph{\textbf{Complexity:}} One can see that the complexity of the SCI component is equal to the complexity of $\mathrm{Zero}$ plus the complexity of secure comparison which is equal to $3L+2$.

\subsection{Secure Comparison Gate}\label{subsection:securecomparisongate}

So far, we have discussed how to compare two secrets using $N$ parties and show the result with a bit. To extend our method for $N$-number secure $\max$ computation, in this subsection, we propose an algorithm, called \emph{Secure Comparison Gate} (\emph{SCG}), which outputs the shares of the partition and $0$-coded vector of the maximum privately instead of sharing just the indicator as in Section \ref{subsection:securecomparisonindicator}. This allows us to combine the secure comparison gates when comparing more than two numbers without revealing anything about the intermediate maximums.
%, and outputs the secret share of partition and $0$-coded vector of $\max(a,b)$, while the privacy is preserved and the nodes do not know anything about the outputs.

The SCG component compares two secret numbers denoted by $a,b \in \mathbb{F}_q$. As shown in Fig. \ref{fig:scg}, SCG is a non-linear component with four input vectors and two output vectors of length $L$. SCG receives partition and $0$-coded vectors of both of $a$ and $b$, and outputs partition vector and $0$-coded vector of $\max(a,b)$ (each node has a share of these). For ease of representation, let us have two sources $A$ and $B$ such that the former has a secret number $a$, and the latter has $b$ as the secret. The goal is to construct the partition vector and $0$-coded vector of $\max(a,b)$ in the output privately. The proposed algorithm is as follows.

%\begin{figure}[htbp]
%	\centering
%	\includegraphics[draft=false,width=90mm]{SCG-1.PNG}
%	
%	\caption{\emph{SCG} is a component that receives partition vectors and $0$-coded vectors of the inputs $a$ and $b$, and outputs partition vector and $0$-coded vector of $\max(a,b)$. }
%	\label{fig:scg}
%\end{figure}

\begin{figure}[htbp]
\centering
\resizebox {.6\columnwidth} {!} {
\begin{tikzpicture}[scale=1]
\tikzstyle{every node}=[font=\small]

\draw[rounded corners,fill=black!30!white] (0,0) rectangle (3,2) node [align=center,pos=.5] {\Large SCG};

\draw[-latex] (3,0.75) -- (3.8,0.75) node at (4.45,0.75) {$\V_{\max(a,b)}^0$};
\draw[-latex] (3,1.25) -- (3.8,1.25) node at (4.45,1.25) {$\V_{\max(a,b)}$};

\draw[-latex] (-1,0.4) -- (0,0.4) node at (-1.35,0.4) {$\V_b^0$};
\draw[-latex] (-1,0.8) -- (0,0.8) node at (-1.35,0.8) {$\V_b$};
\draw[-latex] (-1,1.2) -- (0,1.2) node at (-1.35,1.2) {$\V_a^0$};
\draw[-latex] (-1,1.6) -- (0,1.6) node at (-1.35,1.6) {$\V_a$};

\end{tikzpicture}
}
\caption{SCG is a component that takes partition vectors and $0$-coded vectors of the inputs $a$ and $b$, and produces the partition vector and $0$-coded vector of $\max(a,b)$.}
\label{fig:scg}
\end{figure}
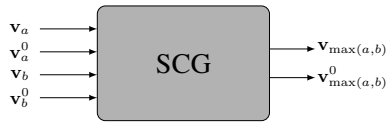

\textbf{Sharing Phase.}
In this phase, $A$ computes $\V_{a}$ and $\V_{a}^0$, and shares them by using Shamir secret sharing. More precisely, $A$ constructs two polynomials $\ab_0(x)=\V_{a}^0+\R^{(1,0)}_1x+\R^{(1,0)}_2x^2+\dots+\R^{(1,0)}_Tx^T$ and $\ab_1(x)=\V_{a}+\R^{(1,1)}_1x+\R^{(1,1)}_2x^2+\dots+\R^{(1,1)}_Tx^T$ , where $\R^{(1,j)}_i$ are chosen uniformly and independently at random from the field $\mathbb{F}_q^{L}$, and sends $\ab_0(\alpha_j)$ and $\ab_1(\alpha_j)$ to each node $j$, $\forall j \in [N]$. Similarly,  $B$ constructs two polynomials $\Bb_0(x)$ and $\Bb_1(x)$ to be able to share $\V_{b}$ and $\V_{b}^0$ and  sends $\Bb_0(\alpha_j)$ and $\Bb_1(\alpha_j)$ to each node $j$, $\forall j \in [N]$.

\textbf{Comparison Indicator.} 
As mentioned in the previous phase, each node $i$ has $\ab_1(\alpha_i)$ and $\Bb_0(\alpha_i)$. Hence, each of the nodes can execute SCI algorithm, as explained in Section \ref{subsection:securecomparisonindicator}, to construct $g(x)$ in collaboration with each other. When the execution is completed, each node $i$ has $g(\alpha_i)$, where $g(0)$ is $0$ if $a>b$, and $g(0)$ is $1$ in other cases.

\textbf{Output Construction.} This is the critical step of the SCG algorithm. In this phase, the goal is to distribute the shares of $\V_{\max(a,b)}$ and $\V_{\max(a,b)}^0$ among the nodes. To do that, the nodes compute $g(0)\V_b+(1-g(0))\V_a$ and $g(0)\V_b^0+(1-g(0))\V_a^0$ together. One can see that if $a>b$, then $g(0)\V_b+(1-g(0))\V_a$ and $g(0)\V_b^0+(1-g(0))\V_a^0$ are equal to $\V_a$ and $\V_a^0$, respectively, and in other cases they are equal to $\V_b$ and $\V_b^0$, respectively. To be able to compute $g(0)\V_b+(1-g(0))\V_a$ and $g(0)\V_b^0+(1-g(0))\V_a^0$, nodes can use MPC as explained in \cite{ben1988completeness,nodehi2019secure}.
Finally, each node $i$ has $\mathbf{o}_1(\alpha_i)$ and $\mathbf{o}_0(\alpha_i)$, where $\mathbf{o}_0(0)=g(0)\V_b^0+(1-g(0))\V_a^0$ and $\mathbf{o}_1(0)=g(0)\V_b+(1-g(0))\V_a$.

After these three phases, the nodes have access to the partition vector and $0$-coded vector of the $\max(a,b)$ while the privacy is preserved. We note that the nodes cannot understand anything about the private inputs and outputs. That is, even though the each node has a share of the partition and $0$-coded vectors of the maximum, they do not know the maximum number. With this, essentially, they continue to compare $\max(a,b)$ with another secret number $c$, which constitutes the basis of our $N$-number secure $\max$ computation scheme.

As for the privacy, the nodes first run the secure comparison indicator (SCI) without revealing anything, which is unconditionally secure as discussed earlier. Further, in the output reconstruction phase, the nodes just do some secure multiplication without revealing any further information. Thus, the overall SCG component does not have any information leakage and it is information-theoretically secure.

\emph{\textbf{Complexity:}} The complexity of the SCG scheme is equal to the summation of complexity of comparison indicator and output construction which is equal to $3L+2 + L + L=5L+2$.

\subsection{$N$-party Secure $\max$ Computation}
\label{subsection:securemaxcomputation}
In this subsection, we propose a scheme to perform the secure $\max$ computation in a system with $N$ nodes. Each node $i$ has a private input $s^{(i)}$. Nodes are interested in computing $\max(s^{(1)},s^{(2)},\dots,s^{(N)})$ without sacrificing the privacy of their secret inputs. For this, first, we partition the nodes into $\ceil{\frac{K}{2}}$ groups of size two, denoted by $\mathcal{G}_1,\mathcal{G}_2,\dots,\mathcal{G}_{\ceil{\frac{K}{2}}}$. Then, the nodes apply the SCG algorithm to find the partition vector and $0$-coded vector of the maximum of each group. Then, they can re-partition the outputs and do it iteratively, until obtaining the partition and $0$-coded vectors of the final result as shown in Fig \ref{fig:maxfunc}. We note that the last entity of partition vector $\V_{\max}$ is equal to $\max$ from (\ref{partition}).

%\begin{figure}[htbp]
%    \centering
%    \includegraphics[draft=false,height=100mm]{maxfunc-1.PNG}
%
%    \caption{By using SCG, nodes can compare secrets, iteratively,  and finally obtain partition vector and $0$-coded vector of the $\max(s^{(1)},s^{(2)},\dots,s^{(N)})$ in collaboration with each other. }
%    \label{fig:maxfunc}
%\end{figure}

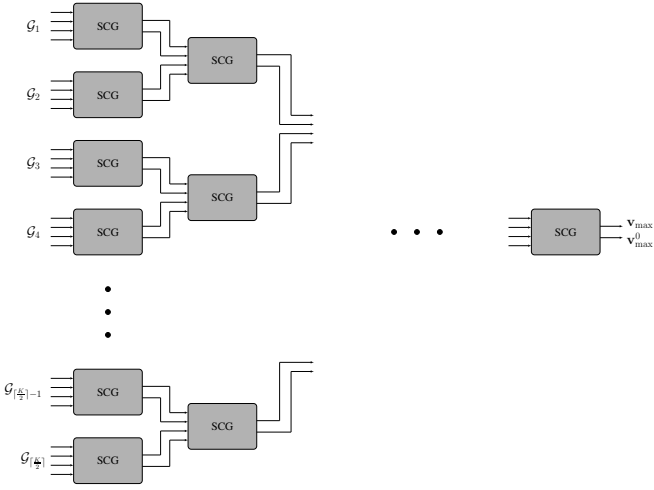
\begin{figure}[htbp]
\centering
\resizebox {.975\columnwidth} {!} {
\begin{tikzpicture}[scale=1]
\tikzstyle{every node}=[font=\LARGE]

\draw[rounded corners,fill=black!30!white] (0,0) rectangle (3,2) node [align=center,pos=.5] {\Large SCG};

%  node at (4.45,0.75) {$\V_{\max(a,b)}^0$}

\draw[-latex] (-1,0.4) -- (0,0.4) node at (-1.75,1) {$\mathcal{G}_1$};
\draw[-latex] (-1,0.8) -- (0,0.8);
\draw[-latex] (-1,1.2) -- (0,1.2);
\draw[-latex] (-1,1.6) -- (0,1.6);

\draw (3,0.75) -- (3.8,0.75);
\draw (3,1.25) -- (4.2,1.25);

%%%%%%%%%%%%%%%%%%%%%%%%%%%%%%%%%%%%%%%%

\draw[rounded corners,fill=black!30!white] (0,-3) rectangle (3,-1) node [align=center,pos=.5] {\Large SCG};

\draw[-latex] (-1,-2.6) -- (0,-2.6) node at (-1.75,-2) {$\mathcal{G}_2$};
\draw[-latex] (-1,-2.2) -- (0,-2.2);
\draw[-latex] (-1,-1.8) -- (0,-1.8);
\draw[-latex] (-1,-1.4) -- (0,-1.4);

\draw (3,-2.25) -- (4.2,-2.25);
\draw (3,-1.75) -- (3.8,-1.75);

%%%%%%%%%%%%%%%%%%%%%%%%%%%%%%%%%%%%%%%%

\draw[rounded corners,fill=black!30!white] (5,-1.5) rectangle (8,0.5) node [align=center,pos=.5] {\Large SCG};

\draw[-latex] (4.2,-1.1) -- (5,-1.1);
\draw[-latex] (3.8,-0.7) -- (5,-0.7);
\draw[-latex] (3.8,-0.3) -- (5,-0.3);
\draw[-latex] (4.2,0.1) -- (5,0.1);

\draw (4.2,0.1) -- (4.2,1.25);
\draw (3.8,-0.3) -- (3.8,0.75);

\draw (4.2,-1.1) -- (4.2,-2.25);
\draw (3.8,-0.7) -- (3.8,-1.75);

%%%%%%%%%%%%%%%%%%%%%%%%%%%%%%%%%%%%%%%%%%
%%%%%%%%%%%%%%%%%%%%%%%%%%%%%%%%%%%%%%%%%%

\draw[rounded corners,fill=black!30!white] (0,-6) rectangle (3,-4) node [align=center,pos=.5] {\Large SCG};

%  node at (4.45,0.75) {$\V_{\max(a,b)}^0$}

\draw[-latex] (-1,-5.6) -- (0,-5.6) node at (-1.75,-5) {$\mathcal{G}_3$};
\draw[-latex] (-1,-5.2) -- (0,-5.2);
\draw[-latex] (-1,-4.8) -- (0,-4.8);
\draw[-latex] (-1,-4.4) -- (0,-4.4);

\draw (3,-5.25) -- (3.8,-5.25);
\draw (3,-4.75) -- (4.2,-4.75);

%%%%%%%%%%%%%%%%%%%%%%%%%%%%%%%%%%%%%%%%

\draw[rounded corners,fill=black!30!white] (0,-9) rectangle (3,-7) node [align=center,pos=.5] {\Large SCG};

\draw[-latex] (-1,-8.6) -- (0,-8.6) node at (-1.75,-8) {$\mathcal{G}_4$};
\draw[-latex] (-1,-8.2) -- (0,-8.2);
\draw[-latex] (-1,-7.8) -- (0,-7.8);
\draw[-latex] (-1,-7.4) -- (0,-7.4);

\draw (3,-8.25) -- (4.2,-8.25);
\draw (3,-7.75) -- (3.8,-7.75);

%%%%%%%%%%%%%%%%%%%%%%%%%%%%%%%%%%%%%%%%

\draw[rounded corners,fill=black!30!white] (5,-7.5) rectangle (8,-5.5) node [align=center,pos=.5] {\Large SCG};

\draw[-latex] (4.2,-7.1) -- (5,-7.1);
\draw[-latex] (3.8,-6.7) -- (5,-6.7);
\draw[-latex] (3.8,-6.3) -- (5,-6.3);
\draw[-latex] (4.2,-5.9) -- (5,-5.9);

\draw (4.2,-5.9) -- (4.2,-4.75);
\draw (3.8,-6.3) -- (3.8,-5.25);

\draw (4.2,-7.1) -- (4.2,-8.25);
\draw (3.8,-6.7) -- (3.8,-7.75);

%%%%%%%%%%%%%%%%%%%%%%%%%%%%%%%%%%%%%%%%

\draw (8,-6.25) -- (9,-6.25);
\draw (8,-6.75) -- (9.5,-6.75);

\draw (8,-0.25) -- (9.5,-0.25);
\draw (8,-0.75) -- (9,-0.75);

\draw (9.5,-0.25) -- (9.5,-2.9);
\draw (9,-0.75) -- (9,-3.3);

\draw (9.5,-6.75) -- (9.5,-4.1);
\draw (9,-6.25) -- (9,-3.7);

\draw[-latex] (9.5,-4.1) -- (10.5,-4.1);
\draw[-latex] (9,-3.7) -- (10.5,-3.7);

\draw[-latex] (9.5,-2.9) -- (10.5,-2.9);
\draw[-latex] (9,-3.3) -- (10.5,-3.3);

%%%%%%%%%%%%%%%%%%%%%%%%%%%%%%%%%%%%%%%%%%
%%%%%%%%%%%%%%%%%%%%%%%%%%%%%%%%%%%%%%%%%%

\draw[rounded corners,fill=black!30!white] (20,-9) rectangle (23,-7) node [align=center,pos=.5] {\Large SCG};

\draw[-latex] (19,-8.6) -- (20,-8.6);
\draw[-latex] (19,-8.2) -- (20,-8.2);
\draw[-latex] (19,-7.8) -- (20,-7.8);
\draw[-latex] (19,-7.4) -- (20,-7.4);

\draw[-latex] (23,-8.25) -- (24,-8.25) node at (24.75,-8.35) {$\V_{\max}^0$};
\draw[-latex] (23,-7.75) -- (24,-7.75) node at (24.75,-7.65) {$\V_{\max}$};

\filldraw (14,-8) circle (3pt);
\filldraw (15,-8) circle (3pt);
\filldraw (16,-8) circle (3pt);

%%%%%%%%%%%%%%%%%%%%%%%%%%%%%%%%%%%%%%%%%%
%%%%%%%%%%%%%%%%%%%%%%%%%%%%%%%%%%%%%%%%%%

\filldraw (1.5,-10.5) circle (3pt);
\filldraw (1.5,-11.5) circle (3pt);
\filldraw (1.5,-12.5) circle (3pt);

\draw[rounded corners,fill=black!30!white] (0,-16) rectangle (3,-14) node [align=center,pos=.5] {\Large SCG};

%  node at (4.45,0.75) {$\V_{\max(a,b)}^0$}

\draw[-latex] (-1,-15.6) -- (0,-15.6) node at (-2.2,-15) {$\mathcal{G}_{\ceil{\frac{K}{2}}-1}$};
\draw[-latex] (-1,-15.2) -- (0,-15.2);
\draw[-latex] (-1,-14.8) -- (0,-14.8);
\draw[-latex] (-1,-14.4) -- (0,-14.4);

\draw (3,-15.25) -- (3.8,-15.25);
\draw (3,-14.75) -- (4.2,-14.75);

%%%%%%%%%%%%%%%%%%%%%%%%%%%%%%%%%%%%%%%%

\draw[rounded corners,fill=black!30!white] (0,-19) rectangle (3,-17) node [align=center,pos=.5] {\Large SCG};

\draw[-latex] (-1,-18.6) -- (0,-18.6) node at (-1.75,-18) {$\mathcal{G}_{\ceil{\frac{K}{2}}}$};
\draw[-latex] (-1,-18.2) -- (0,-18.2);
\draw[-latex] (-1,-17.8) -- (0,-17.8);
\draw[-latex] (-1,-17.4) -- (0,-17.4);

\draw (3,-18.25) -- (4.2,-18.25);
\draw (3,-17.75) -- (3.8,-17.75);

%%%%%%%%%%%%%%%%%%%%%%%%%%%%%%%%%%%%%%%%

\draw[rounded corners,fill=black!30!white] (5,-17.5) rectangle (8,-15.5) node [align=center,pos=.5] {\Large SCG};

\draw[-latex] (4.2,-17.1) -- (5,-17.1);
\draw[-latex] (3.8,-16.7) -- (5,-16.7);
\draw[-latex] (3.8,-16.3) -- (5,-16.3);
\draw[-latex] (4.2,-15.9) -- (5,-15.9);

\draw (4.2,-15.9) -- (4.2,-14.75);
\draw (3.8,-16.3) -- (3.8,-15.25);

\draw (4.2,-17.1) -- (4.2,-18.25);
\draw (3.8,-16.7) -- (3.8,-17.75);

%%%%%%%%%%%%%%%%%%%%%%%%%%%%%%%%%%%%%%%%

\draw (8,-16.25) -- (9,-16.25);
\draw (8,-16.75) -- (9.5,-16.75);

%\draw (8,-10.25) -- (9.5,-10.25);
%\draw (8,-10.75) -- (9,-10.75);

%\draw (9.5,-10.25) -- (9.5,-12.9);
%\draw (9,-10.75) -- (9,-13.3);

\draw (9.5,-16.75) -- (9.5,-14.1);
\draw (9,-16.25) -- (9,-13.7);

\draw[-latex] (9.5,-14.1) -- (10.5,-14.1);
\draw[-latex] (9,-13.7) -- (10.5,-13.7);

%\draw[-latex] (9.5,-12.9) -- (10.5,-12.9);
%\draw[-latex] (9,-13.3) -- (10.5,-13.3);

\end{tikzpicture}
}
\caption{Using SCG, nodes can iteratively compare secrets and ultimately collaboratively obtain the partition vector and '0'-coded vector of $\max\left(s^{(1)},s^{(2)},\dots,s^{(N)}\right)$.}
\label{fig:maxfunc}
\end{figure}

As for the privacy, it is shown that each SCG is unconditionally secure. Thus, designing  a circuit based on SCGs does not leak any information and is unconditionally secure.

\emph{\textbf{Complexity:}} The proposed  algorithm is for finding the maximum of $N$ secret numbers without any information leakage. In the $\max$ function algorithm, the complexity is $N-1$ times of complexity of SCG. Thus, the total complexity is $(N-1)(5L+2)$.

\begin{Remark}
Proposed schemes in \cite{damgaard2006unconditionally,nishide2007multiparty,makri2021Rabbit} are only for the comparison of two secret numbers that are shared among $N$ parties. To be able to apply the schemes in \cite{damgaard2006unconditionally,nishide2007multiparty,makri2021Rabbit} to the comparison of more than two secret numbers, one needs to do iterative pairwise comparison which leaks the intermediate maximums. Further, all of these schemes in \cite{damgaard2006unconditionally,nishide2007multiparty,makri2021Rabbit} require higher computation complexity than our scheme, as shown in Tables~\ref{table:comparison} and~\ref{table:equality}.
\end{Remark}

\begin{Remark}
\label{remark:mincomputatio}
It is straightforward to modify the proposed method to compute minimum of $N$ secret numbers. To compute the minimum privately, first the nodes multiply all shares by $-1$ and then apply secure $\max$ function on the secrets, and finally change the sign of the result.
\end{Remark}

\section{Applications}
\label{section:application}

In this section, we discuss some applications of the proposed secure comparison algorithm in Section \ref{section:proposedcomparison}.

\subsection{Secure Auction}
\label{subsection:secureauction}

One of the main uses of $N$-party comparison is computing the maximum of $N$ secret numbers, which is widely used in auctions \cite{mcmillan1994selling,rassenti1982combinatorial}. In \cite{lehmann2002truth}, a solution for auction is proposed by using an approximate greedy algorithm. The downside of this scheme is that the bids are not private. A secure solution for auction was proposed in \cite{franklin1996design} that uses multiple servers as trusted third parties and the bid topology is revealed to these trusted third parties. Micali and Rabin proposed a method in \cite{micali2014cryptography} that preserves the privacy of the bid by using Pedersen commitment. At the end of the bidding phase, this protocol reveals the bid information to the auctioneers. Another practical and multi-party auction protocol was presented in \cite{brandt2005efficient} which is computationally secure. 

By using the proposed algorithm in Section \ref{subsection:securemaxcomputation}, one can find the maximum price in an auction in a completely information-theoretically private manner without revealing any of the bids. In addition, if the nodes aim to deduce the index of the maximum price to determine the node which offers the maximum price without even revealing any information about the maximum bid, one can extend SCG such that it outputs the index of the maximum, i.e., the node which makes the greatest bid, similar to the process of outputting the partition and $0$-coded vectors of the maximum. For clarification, we detail the extended version of SCG in the following.

\subsubsection{Extended SCG}

In Section \ref{subsection:securecomparisongate}, we show how to design the SCG component for comparing two secret inputs $a$ and $b$ such that it outputs the secret share of the partition and $0$-coded vectors of $\max(a,b)$, while the privacy is preserved and the nodes do not know anything about the inputs and outputs beyond their shares. Here, we aim to modify the SCG design and propose \emph{Extended SCG (ESCG)}. As shown in Fig. \ref{fig:escg}, ESCG receives the partition and $0$-coded vectors of the secret numbers $a$ and $b$ just like the original SCG along with indices of $a$ and $b$. ESCG outputs the partition and $0$-coded vectors of $\max(a,b)$ as well as the index of the $\max(a,b)$. By using ESCG instead of SCG in Section \ref{subsection:securemaxcomputation}, the nodes also can compute the index of the maximum, which is useful in secure auction in determining the winner of the auction, without revealing the individual bids.

%\begin{figure}[htbp]
%	\centering
%	\includegraphics[draft=false,width=60mm]{ESCG-1.PNG}
%	\caption{\emph{ESCG} receives partition vector, $0$-coded vector and index of each of $a$ and $b$, and outputs partition vector, $0$-coded vector and index of the $\max(a,b)$. }
%	\label{fig:escg}
%\end{figure}

\begin{figure}[htbp]
\centering
\resizebox {.6\columnwidth} {!} {
\begin{tikzpicture}[scale=1]
\tikzstyle{every node}=[font=\small]

\draw[rounded corners,fill=black!30!white] (0,-0.5) rectangle (4,2.5) node [align=center,pos=.5] {\Large ESCG};

\draw[-latex] (4,1) -- (4.8,1) node at (5.45,1) {$\V_{\max(a,b)}^0$};
\draw[-latex] (4,1.5) -- (4.8,1.5) node at (5.45,1.5) {$\V_{\max(a,b)}$};
\draw[-latex] (4,.5) -- (4.8,.5) node at (5.45,.5) {$i_{\max(a,b)}$};

\draw[-latex] (-1,0.7) -- (0,0.7) node at (-1.35,0.7) {$\V_b$};
\draw[-latex] (-1,0.3) -- (0,0.3) node at (-1.35,0.3) {$\V_b^0$};
\draw[-latex] (-1,-0.1) -- (0,-0.1) node at (-1.35,-0.1) {$i_b$};

\draw[-latex] (-1,2.1) -- (0,2.1) node at (-1.35,2.1) {$\V_a$};
\draw[-latex] (-1,1.7) -- (0,1.7) node at (-1.35,1.7) {$\V_a^0$};
\draw[-latex] (-1,1.3) -- (0,1.3) node at (-1.35,1.3) {$i_a$};

\end{tikzpicture}
}
\caption{ESCG receives the partition vector, $0$-coded vector, and index of each of $a$ and $b$, and outputs the partition vector, $0$-coded vector, and index of $\max(a,b)$.}
\label{fig:escg}
\end{figure}
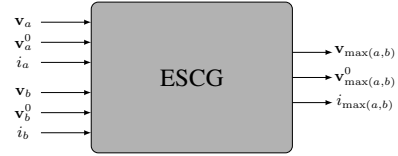

%As shown in , SCG is a non-linear component with six input vectors and three outputs. In other words, SCG receives partition and $0$-coded vectors and index of both of $a$ and $b$, and outputs partition vector and $0$-coded vector and index of $\max(a,b)$. 
For ease of presentation, we assume that the partition vector, $0$-coded vector and index of two secret numbers $a$ and $b$ are shared among $N$ nodes. The goal is to construct the partition vector, $0$-coded vector and index of $\max(a,b)$ in the output privately. The proposed algorithm largely follows from that of Section $\ref{subsection:securecomparisongate}$, except that in the output construction the goal is to distribute the shares of $\V_{\max(a,b)}$, $\V_{\max(a,b)}^0$ and $i_{\max(a,b)}$ among the nodes. Obtaining the shares of $\V_{\max(a,b)}$ and $\V_{\max(a,b)}^0$ is the same as Section \ref{subsection:securecomparisongate}. To gain the shares of $i_{\max(a,b)}$, the nodes compute $g(0)i_b+(1-g(0))i_a$ together. One can see that if $a>b$, then $g(0)i_b+(1-g(0))i_a$ is equal to $i_a$, and in other cases it is equal to $i_b$.

Thus, in the end, the nodes have access to the partition vector, $0$-coded vector and index of $\max(a,b)$ while the privacy is preserved. We note that the nodes cannot understand anything about the private inputs and outputs as each of them only has access to a secret share.

\subsection{Secure Median Computation}
\label{subsection:median}
Median is one of the main operations in distributed algorithms. It is shown that median-based distributed gradient descent algorithm is robust against outliers and certain attacks, especially untargeted poisoning\cite{yin2018byzantine}. Further, median-based distributed gradient descent can be combined with different methods to increase the convergence rate in learning.
In this part, the setting is the same as Section \ref{section:problemsetting}, except that the nodes are interested in finding the median of their secret inputs $s_1,s_2,\ldots,s_N$ as opposed to computing the $\max$ among them.

\begin{enumerate}
	\item To compute the median, first, all of the nodes compare each pair of $s_1,s_2,\ldots,s_N$ by using SCI as described in Section \ref{subsection:securecomparisonindicator}. In other words, for each $s_i$, the nodes run secure comparison indicator algorithm and compare it with each of the other inputs $s_1,s_2,\ldots,s_N$.
	
	\item Up to now, each node $j$ has $g_1(\alpha_j),g_2(\alpha_j),\dots,g_N(\alpha_j)$, where $g_k(0)$ is equal to $0$, if  $s_i$ is greater than $s_k$, and it is equal to 1 in other cases, $\forall k \in [N]$. Let us define $g(x)=\sum_{k=1}^Ng_k(x)$. Hence, each node $j$ can compute $g(\alpha_j)=\sum_{k=1}^Ng_k(\alpha_j)$, where $g(0)$ is equal to $\frac{N}{2}$, if and only if $s_i$ is the median.
	
	\item To verify that $g(0)$ is equal to $\frac{N}{2}$,
	the nodes can use equality test as mentioned in Section \ref{subsection:equality and zero indicator}. If $g(0) \neq \frac{N}{2}$, the nodes do the same with another secret $s_i$.
	
\end{enumerate}

\begin{Remark}
	One can see that we can extend this algorithm to find the secret with special rank among all of the other secrets, which is shown by $\mathrm{Rank}_t$.	
	This function returns the secret with statistical rank of $t$. To be able to compute $\mathrm{Rank}_t$, we can do the same as median, except that in the end we verify whether $g(0)-t$ is equal to $0$ or not. By doing so, we essentially order the secrets without knowing their individual values.
\end{Remark}

\subsection{Computing Non-Polynomial Functions of the Inputs}
\label{subsection:computingnonlinearfunctions}
We note that in our proposed comparison scheme in Section~\ref{section:proposedcomparison} and its applications discussed in this section, each of the nodes has a share of the result before the final revealing phase. This allows the nodes to compute many non-polynomial functions of the inputs including $\max$, $\min$, $\mathrm{median}$, $\mathrm{Rank}_t$, etc as discussed earlier. In addition, having shares of the result in the final revealing phase also enables nodes to continue the computations to compute other functions based on the comparison result. For example, in Section \ref{subsection:securemaxcomputation}, each node derives a share of the partition vector and $0$-coded vector of the maximum. The last entity of partition vector $\V_{\max}$ is equal to $\max$. So each node has a share of the maximum, and by using that, the nodes can compute any polynomial function of the maximum. Here, we present some such scenarios.

%For example, assume that we have $5$ nodes with private secrets $S_1, S_2, \ldots, S_5$. By using the proposed algorithm, we can compute $\max(S_1,S_2,S_4)^2+\mathrm{median}(S_1,S_2, \ldots, S_5)+S_1S_2$. 
\subsubsection{Outlier detection}

One of the main challenges in many applications such as federated learning (FL)  or model combination is to remove the outlier models.
In an FL system, some data owners train a local model and send the local parameters to a central server which is interested in updating a global model by using these local models. Usually in FL systems, there are some data owners whose models are not good enough potentially due to low computation power or low quality of private data. In \cite{chen2017distributed}, it is shown that only one outlier model can influence the final result and prevent convergence in FL systems.

As shown in \cite{chen2017distributed}, using the average of the local models as an aggregation function at the central server is not an outlier-resistance solution, as it can be influenced by the outliers. On the other hand, in \cite{pillutla2022robust,yin2018byzantine}, it is shown that $\mathrm{median}$ of the local models is robust against a certain number of outliers. Thus, one of the straightforward and effective algorithms to combat outliers is to remove the data which has large distance from the median \cite{guerraoui2018hidden}.

Assume that there are $K$ clients numbered $1,2,...,K$, such that each clients $i$ has access to a private value $x_i \in \mathbb{F}_q$, $\forall i \in [K]$. These private values can be gradient vectors in federated learning or private models in a model combination algorithm. When the central server is interested in detecting the outlier values to remove them in the aggregation phase through the median approach, it needs to compute the distances $d_i=|x_i - \mathrm{median}(\{x_j\}_{j=1}^K)|^2$, $\forall i \in [K]$, without revealing anything to the clients. To do that, the clients compute the shared values of $\mathrm{median}(\{x_j\}_{j=1}^K)$ in collaboration with each other as explained in Section \ref{subsection:median}. It is shown that the last entity of partition vector of the median is the exact value of median. So, each node $i$ has a share of the median in the form of Shamir secret sharing, and by using BGW scheme, they can derive the shares $(x_i - \mathrm{median}(\{x_j\}_{j=1}^K))(x_i - \mathrm{median}(\{x_j\}_{j=1}^K))$.
Finally, the clients reveal the shares of $d_1,d_2,...,d_K$ just to the central server, so that the central server can compute the exact values of the distances $d_1,d_2,...,d_K$. With that, the central server can remove the outliers during the aggregation phase, without seeing the individual models of the clients.

%\begin{Remark}
%In this scheme, we assume that a certain number of the nodes are semi-honest. But we can 
%\end{Remark}

\subsubsection{Minimax \& MaximinM Functions}
\label{subsubsection:minimax}

Another useful function that we can privately evaluate using the proposed algorithm is the minimax function, which has applications in decision theory, learning, and statistics. Let us assume that we have a total of $K$ user groups, each with $N_k$ users, $k \in [K]$. User $j$ of group $k$ holds a secret value $s_{j,k}$, $j \in [N_k]$. Users want to find out in which group the minimum secret input is the largest. %For example, groups can represent distinct coverage areas and the secret inputs of the users can represent the number of times each user interacted with their mobile phones. 
That is, the goal is to compute the following function in a completely private manner without revealing any information about the secrets. 
%Consider $n_1$ students from school $S_1$, $n_2$ students from school $S_2$, ...., and $n_K$ students from school $S_K$. Assume that all of the students are interested in finding the school with the best education system. The comparison parameter is that in which of the schools the minimum GPA is higher than the others. Let $s_{i,j}$ denote the $i$th student from school $S_j$, and $g_{i,j}$ is the GPA of $s_{i,j}$. Hence, the goal is to compute the following minimax function in a completely private manner without revealing any information about the GPAs.
\begin{align}\label{maxmin}
\max_{k\in[K]}\min_{j\in[N_k]}s_{j,k}
\end{align}

To reach the goal, in the first phase, by using Shamir secret sharing the users share their secret inputs. Then, all of the users run the secure $\min$ computation, based on Remark \ref{remark:mincomputatio}, over the inputs of the users from the same group. After this phase, each user has a share of partition vector and $0$-coded vector of the minimum of each group. Thus, they can run secure $\max$ computation to find the maximum of the results as well as the group index of the maximum as explained in Section \ref{subsection:secureauction}.
Hence, the users can compute (\ref{maxmin}) without any information leakage about their secret inputs. We note that the proposed secure comparison technique can be utilized to compute other such functions including minimax or max of median.

\section{Conclusion}
\label{section:conclusion}

In this paper, we initially introduced a novel coding method known as the $0$-coding vector, enabling nodes to perform comparisons with low complexity. Subsequently, we proposed a method for equality testing. By combining these techniques, we developed an algorithm capable of computing the maximum of multiple numbers. We then explore the integration of these schemes for applications in secure electronic auctions or the computation of non-linear functions, such as median, min/max functions, rank, etc.

\bibliographystyle{ieeetr}
\bibliography{References}

\end{document}